\documentclass[letterpaper, 10 pt, letterpaper]{ieeeconf}

\usepackage{cite}
\usepackage{amsmath,amssymb,amsfonts}
\usepackage{graphicx}
\usepackage{textcomp}
\usepackage{multirow}
\usepackage{algorithm,algorithmic}
\usepackage{todonotes}
\usepackage{xargs}
\newcommandx{\iman}[2][1=]{\todo[linecolor=orange,backgroundcolor=orange!25,bordercolor=orange,author=Iman,#1]{#2}}

\newcommandx{\nm}[2][1=]{\todo[linecolor=orange,backgroundcolor=red!25,bordercolor=red,author=Nader,#1]{#2}}

\usepackage{lineno,hyperref,mathrsfs,epstopdf,systeme}

\newtheorem{theorem}{\textbf{Theorem}}
\newtheorem{proposition}{\textbf{Proposition}}
\newtheorem{assumption}{\textbf{Assumption}}
\newtheorem{definition}{\textbf{Definition}}
\newtheorem{remark}{\textbf{Remark}}
\allowdisplaybreaks
\renewcommand{\baselinestretch}{0.86}

\begin{document}
	\title{Cyber Attack and Machine Induced Fault Detection and Isolation Methodologies for  Cyber-Physical Systems}
	\author{Mahdi Taheri, Khashayar Khorasani, Iman Shames, and Nader Meskin
		\thanks{The authors would like to acknowledge the financial support received from NATO under the Emerging Security Challenges Division program. K. Khorasani and N. Meskin would like to acknowledge the support received from NPRP grant number 10-0105-17017 from the Qatar National Research Fund (a member of Qatar Foundation). K. Khorasani would also like to acknowledge the support received from the Natural Sciences and Engineering Research Council of Canada (NSERC) and the Department of National Defence (DND) under the Discovery Grant and DND Supplemental Programs. The statements made herein are solely the responsibility of the authors.}
		\thanks{Mahdi Taheri (m$\_$eri@encs.concordia.ca) and Khashayar Khorasani (kash@ece.concordia.ca) are with the Department of Electrical and Computer Engineering, Concordia University, Montreal, Canada.}%
		\thanks{Iman Shames (iman.shames@unimelb.edu.au) is with the Department of Electrical and Electronic Engineering, University of Melbourne, Melbourne, Australia.}%
		\thanks{ Nader Meskin (nader.meskin@qu.edu.qa) is with the Department of Electrical Engineering, Qatar University, Doha, Qatar.}%
	}
	\maketitle
	
	\begin{abstract}
		In this paper, the problem of simultaneous cyber attack and fault detection and isolation (CAFDI) in cyber-physical systems (CPS) is studied. The proposed solution methodology consists of two filters on the plant and the command and control (C\&C) sides of the CPS and an unknown input observer (UIO) based detector on the plant side. Conditions under which the proposed methodology can detect deception attacks, such as covert attacks, zero dynamics attacks, and replay attacks are characterized. An advantage of the proposed methodology is that one does not require a fully secured communication link which implies that the communication link can be compromised by the adversary while it is used to transmit the C\&C side observer estimates. Also, it is assumed that adversaries have access to parameters of the system, filters, and the UIO-based detector, however, they do not have access to all the communication link channels. Conditions under which, using the communication link cyber attacks, the adversary cannot eliminate the impact of actuator and sensor cyber attacks are investigated. To illustrate the capabilities and effectiveness of the proposed CAFDI methodologies,  simulation case studies are provided and comparisons with detection methods that are available in the literature are included to demonstrate the advantages and benefits of our proposed solutions.
	\end{abstract}
	
%	\begin{IEEEkeywords}
%		Cyber attack and fault detection and isolation (CAFDI), cyber-physical systems (CPS), actuator attack, sensor attack, covert attack, zero dynamics attack, unknown input observer (UIO).
%	\end{IEEEkeywords}

	\section{Introduction}
	Cyber-physical systems (CPS) are monitored and controlled by distributed sensors, actuators, and embedded computers that are connected via communication networks \cite{ascf}. Our today's life massively depends on CPS due to their wide range of applications in different areas, such as power systems and smart grid, next generation aerospace and transportation systems, and process control and water treatment networks \cite{adaii}. Through employing CPS for these applications provide us with unique capabilities to accomplish high level performance and reliability  performing complex tasks \cite{DTA}.

	Anomalies and machine induced faults as well as malicious cyber attacks in physical components of CPS do occur and are observed in actuators and sensors. In recent years, cyber security challenges in CPS, that include cyber attacks on communication networks have attracted significant interest \cite{7581101,5394956,6760152,adaii,DTA,asorr}. Nevertheless, the problem of \textit{simultaneous} diagnosis of  cyber attacks and faults  has not been fully addressed in the literature.

A special type of cyber attack is defined as the deception attack in which an adversary changes the transmitted information of the system's input or output  by compromising the CPS network communication channels. This paper studies the cyber attack and fault detection and isolation (CAFDI) problem of CPS in  presence of machine induced faults as well as malicious deception cyber attacks, such as covert attacks, zero dynamics attacks, and replay attacks. Covert attacks and zero dynamics attacks are defined as undetectable attacks \cite{safcpsuua,siolcps,baniamerian2019determination}, since they have no impact on the received output measurements on the command and control (C\&C) side of the CPS.

	A number of researchers have attempted to directly apply fault detection methods to detect cyber attacks, however, there is an inherent difference between machine induced faults and cyber attacks anomalies. Faults represent structural physical  anomalies in the system, whereas cyber attacks are injected intentionally by an intelligent adversary with the purpose of damaging the nominal behavior of the system. Standard fault detection algorithms, such as unknown input observer (UIO) \cite{NCSU}, have been used as tools to detect cyber attacks. There is an inherent differences between faults and cyber attacks, where faults follow and are governed by laws of physics and are associated with physical system properties. On the other  hand,  cyber attacks  are intelligently designed and do not necessarily follow physical system degradations. Consequently, conventional fault diagnosis algorithms should be fundamentally generalized to accommodate the malicious intelligent adversary cyber attacks threats.

	As a brief overview, the  geometric-based fault detection methodologies  were proposed in \cite{mmit,gsot} to obtain necessary and sufficient conditions for existence of observers that can be used to generate a residual signal for the purpose of fault detection and isolation (FDI). In addition to geometric approaches, many algebraic model-based FDI methods have been introduced, such as UIO \cite{douioarfdf,sfdvro}, interacting multiple model \cite{1507309}, multiple model \cite{sfdiaiummbhkffgt,ammbaffdoje}, distributed detection algorithms \cite{MESKIN20092032,DAVOODI2016185,dfdfi}, and parity equation based approaches \cite{fdaiupr,aropsatfd}.

	For the cyber attack detection problem, a periodic modulation scheme with the idea of changing the behavior of the control input was proposed in \cite{docaazdaicps} to detect covert and zero dynamics attacks in CPS. However, by using this method a fault in the system can misleadingly be detected as a cyber attack. A method to detect covert attacks in a network of interconnected subsystems using the received information from subsystems was introduced in \cite{barboni2019distributed}. However, it was assumed that the communication links among the subsystems are \textit{fully secured}, which is not always feasible in real-world systems.

	In \cite{rsaic}  geometric theory was used to define zero dynamics attacks and show their impact on the system, and proposed to add perturbations to the system matrices of the system $(A,B,C)$ to change the zero dynamics of the system so that the adversary can no longer excite these new zero dynamics modes. However, in a zero dynamics cyber attack, the adversary has a complete knowledge of the system, therefore, after changing the characteristics of the system one would still be able to discover the new matrices and dynamics.

	In \cite{csfsc}, a sensor coding method was proposed that reveals stealthy false data injection attacks by changing the direction of cyber attacks where an algorithm to compute the coding matrices was designed, and finally, a time-varying coding approach was developed for the case when the adversary is capable of estimating a static coding matrix. As a drawback of this approach, it should be noted that one is also not capable of isolating faults and cyber attack signals and anomalies.

	The authors in \cite{diaocsuata} developed a moving target approach in which certain time-varying external dynamics are added to the system. Leveraging the moving target approach, the extended dynamics of the system become unknown to adversaries and they no longer are capable of executing covert attacks and replay attacks. However, zero dynamics attacks cannot be detected by using the moving target approach. In \cite{docaocpsbe}, the system was augmented by adding switching auxiliary dynamics that are unknown to the adversary and a switched Luenberger observer was designed to detect covert and zero dynamics attacks, however, for implementation purposes the extended system and the switched observer need to be synchronized.

	Due to stealthiness of covert and zero dynamics attacks, it is of paramount importance to develop methods that can be used to detect and isolate them. In addition, due to existence of physical component faults in CPS, one needs to also clearly detect and isolate both faults and cyber attacks in these systems. This paper aims at addressing the problem of CAFDI in CPS.

In our proposed methodology, two filters are designed on both the plant side and the C\&C side of the CPS that are interconnected  via communication links that can be compromised by the adversary. Moreover, on the plant side UIO-based detectors are designed to generate residuals for detecting and isolating actuator cyber attacks, sensor cyber attacks, as well as actuator faults, and sensor faults while the adversary have a \underline{complete knowledge} of the filters and  UIO-based detectors. Any type of detectable and undetectable cyber attacks can be detected by using our proposed methodology, however, we have assumed that the adversary does not have access to all the communication  channels among the filters. 

By utilizing both the filters and  detectors, we propose and derive conditions under which an adversary that performs cyber attack on the communication link channels cannot eliminate the impacts of actuator and sensor attacks.

	To summarize, the main contributions of this paper are stated as follows:
	
	\begin{enumerate}
		\item A distributed filter design methodology based on observing the system from both the plant side and the C\&C side is introduced and developed that can be utilized to detect and isolate both cyber attacks and machine induced faults.
		
		\item By utilizing our proposed methodology, undetectable cyber attacks such as covert attacks and zero dynamics attacks, as well as detectable attacks such as replay attacks can be detected and isolated.
		
		\item Based on both the plant side and the C\&C side estimation and observation methodology, conditions under which isolation among actuator cyber attacks and sensor cyber attacks are provided and developed.
	\end{enumerate}
	
	The remainder of the paper is organized as follows. A mathematical model of the system that takes into account faults and cyber attacks, the definition of undetectable attacks, and the main objective of this paper are provided in Section \ref{s:models}. In Section \ref{s:ObserverD}, our proposed CAFDI methodology that consists of two side filters, the UIO-based detector and residual signals are developed and investigated. Design conditions for the filters and  detector are proposed and developed. To illustrate and demonstrate the capabilities of our analytical results,  numerical simulation case studies are presented in Section \ref{s:simo}. Conclusions are provided in Section \ref{s:conclu}.

	\section{Problem Statement and Formulation}\label{s:models}
	
	\subsection{The Cyber-Physical System (CPS) Model}
	
	In this paper, a strictly proper linear time-invariant (LTI) CPS of the form given below is studied:
	
	\begin{align}\label{e1}
	\dot{x}^\text{s} (t)=&A^\text{s} x^\text{s} (t)+B^\text{s} u^*(t)+L_1f_1(t) + N^\text{s}  \omega^\text{s}  (t),\nonumber\\
	y_\text{p}(t)=& C^\text{s} x^\text{s} (t)+L_2f_2^\text{s} (t) + \nu^\text{s}  (t),
	\end{align}
	
	\noindent where $x^\text{s} (t)\in \mathbb{R}^n$ represents the state, $y_\text{p}(t)\in \mathbb{R}^p$ denotes the measured output on the plant side, $u^*(t)\in \mathbb{R}^m$ denotes the control input, $f_1(t)\in \mathbb{R}^{m_\text{f}}$ and $f_2^\text{s} (t)\in \mathbb{R}^{p_\text{f}}$ correspond to actuator and sensor faults, respectively. Moreover, $\omega^\text{s}  (t) \in \mathbb{R}^m$ and $\nu^\text{s}  (t) \in \mathbb{R}^p$ denote zero mean wide-sense stationary (WSS) random Gaussian processes that represent process and measurement noise with the covariance matrices $Q$ and $R$, respectively. The quadruple   $(A^\text{s} ,\,C^\text{s} ,\, B^\text{s} , \, N^\text{s} )$ has appropriate dimensions and describe the CPS characteristics, and the known pair $(L_1,\,L_2)$ capture the fault signatures.

	In case of injection of a cyber attack on actuators, the control input is expressed and changed to
	\begin{equation}\label{e:c_i_a}
	u^*(t)=u(t)+S_{\text{a}}a_{\text{u}}(t),
	\end{equation}
	where $u(t) \in \mathbb{R}^m$ represents the control command which is the output of the C\&C, $a_{\text{u}} (t) \in \mathbb{R}^{m_\text{a}}$ denotes a vector describing the effects of unknown cyber attacks on actuators, and $S_{\text{a}}$ is a matrix of appropriate dimension which indicates the control input channels that are under attack.

	The output of the CPS on the C\&C side when sensors are under cyber attack can be expressed as
	\begin{equation}\label{e:o_a}
	y^*(t)=C^\text{s} x^\text{s} (t)+L_2 f_2^\text{s} (t)+D_{\text{a}}a_\text{y}(t)+\nu^\text{s} (t),
	\end{equation}
	where $y^*(t) \in \mathbb{R}^p$ denotes the output, $a_y(t) \in \mathbb{R}^{p_\text{a}}$ denotes the attack signal, and the known matrix $D_{a}$ describes the sensor attack signature. A CPS in  presence of both the actuator and sensor cyber attacks is depicted in Fig.~\ref{fig:cps}.

	Equations \eqref{e1} and \eqref{e:c_i_a} provide a state space realization of the CPS from the C\&C side in the following form:
	\begin{align}
	\dot{x}^\text{s}(t) =& A^\text{s} x^\text{s}(t)+B^\text{s} u(t)+B_{\text{a}}^\text{s} a_\text{u}(t)+L_{1}f_{1}(t) + N^\text{s} \omega^\text{s} (t),\label{e:e_s_p}
	\end{align}
	where $B_{\text{a}}^\text{s}=B^\text{s} S_{\text{a}}$ is to be interpreted as the actuator cyber attack signature.

	In \eqref{e:c_i_a} and \eqref{e:o_a}, $a_\text{u} (t)$ and $a_\text{y} (t)$ denote the impacts of the adversary's attack on the control input and output of the CPS, respectively. The signals $a_\text{u} (t)$ and $a_\text{y} (t)$ can be arbitrarily changed by the malicious adversary. In  presence of $a_\text{u} (t)$ and $a_\text{y} (t)$, the adversary intends to inflict maximum possible damage on the components of the system while simultaneously remaining undetected. 	The following definitions are needed in the remainder of the paper.

	\begin{figure}[!t]
		\centering
		\centerline{\includegraphics[width=\columnwidth]{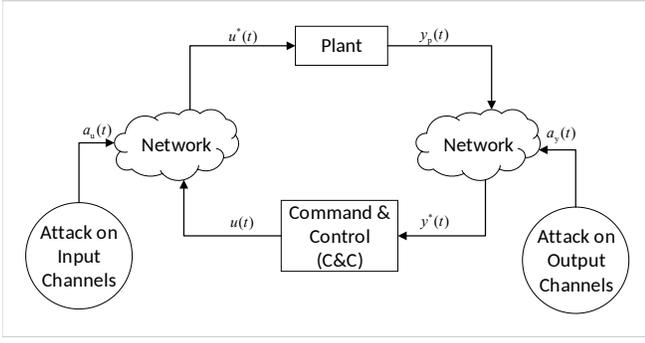}}
		\caption{Cyber-physical system under deception attack on both input and output channels, where $u(t)$ denotes the control command, $a_\text{u} (t)$ represents the cyber attack signal on the input channel, $u^*(t)$ represents the control input of the plant, $y_\text{p}(t)$ denotes the output on the plant side, $a_\text{y} (t)$ denotes the attack signal on the output channel, and $y^*(t)$ denotes the output on the C\&C side.}\label{fig:cps}
	\end{figure}

	\begin{definition}[Weakly Unobservable Subspace \cite{ctfls}] \label{def:weaklyObs}
		Let us denote the CPS by $\Sigma=(A^\text{s} ,B^\text{s} ,B_a^\text{s} ,L_1,N^\text{s} ,C^\text{s} ,L_2, D_\text{a})$. Under the fault free scenario $f_1(t)=0$ and $f_2^\text{s} (t)=0$, the noise free scenario $\omega^\text{s}  (t) =0$ and $\nu^\text{s}  (t) =0$, and the cyber attack free scenario $a_\text{u}(t)=0$ and $a_\text{y}(t)=0$, a point $x^\text{s} (0)=x_0^\text{s}  \in \mathbb{R}^n$ is called weakly unobservable if there exists an input function $u(t)$ such that the output satisfies $y^*(t)=0$, $\forall \, t \geq 0$. The set of all weakly unobservable points is called weakly unobservable subspace and is denoted by $\mathscr{V}(\Sigma)$.
	\end{definition}
	
Let us denote $X^\text{s} (x^\text{s} (0),u(t),a_\text{u}(t),a_\text{y}(t))$ as the solution to \eqref{e:e_s_p} under the fault free condition, and $Y(x^\text{s} (0),u(t),a_\text{u}(t),a_\text{y}(t))=C^\text{s}  X^\text{s} (x^\text{s} (0),u(t),a_\text{u}(t),a_\text{y}(t))$ as the corresponding output of the CPS, $\forall \, t \geq 0$.
	\begin{definition} [Undetectable Cyber Attacks \cite{siolcps}] \label{def:undetectable}
		Given $x^\text{s} (0) =x_0^\text{s}$, in the CPS (\ref{e:e_s_p}) under the fault free scenario, the cyber attack on actuators and sensors using $a_\text{u} (t) \neq 0$ and $a_\text{y} (t)$, is designated as undetectable if $Y(x_0^\text{s} ,u(t),a_\text{u}(t),a_\text{y}(t))=Y(x_0^\text{s},u(t),0,0)$, $\forall t \geq 0$.
	\end{definition}
	% 	\begin{definition} [Input Observable Systems\cite{massoumnia1989failure}] \label{def:inputObs}
	% 		The system $(C,\, A, \, B)$ is input observable if $B$ is monic and $\text{Im}(B)$ does not intersect with the unobservable subspace of $(C, \, A)$ where $\text{Im}(B)$ denotes the image of $B$.
	% 	\end{definition}

	In the same manner as described in \cite{gsot, massoumnia1989failure}, the sensor fault and sensor noise can be represented by pseudo actuator fault and pseudo process noise, respectively. It is worth noting that in this representation, as described below, sensor faults are mapped into and represented by pseudo actuator faults. 

Towards the above end, the following auxiliary invertible LTI system that is driven by the appropriate $f_2 (t)$, which represents the pseudo actuator fault, and $\omega^\text{a} (t)$, which captures the pseudo process noise, is expressed as:
	\begin{equation}\label{e:psedu}
	\begin{split}
	\dot{x}^\text{a} (t) &=A^\text{a} x^\text{a} (t)+L_2^\text{a} f_2 (t)+N^\text{a} \omega^\text{a} (t), \\
	C^\text{a} x^\text{a} (t) &= L_{2}f_2^\text{s} (t) +\nu^\text{s} (t),
	\end{split}
	\end{equation}
	where $x^\text{a} (t) \in \mathbb{R}^{p_\text{f}+{p}}$, $f_2 (t) \in \mathbb{R}^{p_\text{f}}$, and $\omega^\text{a} (t) \in \mathbb{R}^{p}$. By incorporating the dynamics of \eqref{e:e_s_p} and \eqref{e:psedu}, one can obtain the augmented and extended CPS in the following form:
	\begin{align}
	\dot{x}(t) =& A x(t)+Bu(t)+B_{\text{a}}a_\text{u}(t)+F_{1}f_{1}(t)+F_2 f_2 (t) \nonumber \\
	& + N \omega (t),  \nonumber \\
	y^*(t) =& C x(t)+D_{\text{a}}a_\text{y}(t),\label{e:e_s}
	\end{align}
	where $x(t)=[x^\text{s} (t)^\top , \, x^\text{a} (t)^\top]^\top$, $A=\text{diag}(A^\text{s},A^\text{a})$, $B=[{B^\text{s}}^\top ,\, 0_{m \times (p_\text{f}+p)}]^\top$, $B_\text{a}=[{B_\text{a}^\text{s}}^\top ,\, 0_{m_\text{a} \times (p_\text{f}+p)}]^\top$, $F_1=[{L_1}^\top ,\, 0_{m_\text{f} \times (p_\text{f}+p)}]^\top$, $F_2=[0_{p_\text{f} \times n} ,\, {L_2^\text{a}}^\top]^\top$, $N=\text{diag}(N^\text{s}, N^\text{a})$, $\omega (t)=[\omega^\text{s} (t)^\top ,\, \omega^\text{a} (t)^\top]^\top$, and $C=[C^\text{s}, \, C^\text{a}]$. It should be noted that the defined output  $y^*(t)$ in \eqref{e:o_a} is equal to the one that is given by \eqref{e:e_s}, however, the representations are different.
	
	\subsection{Objectives}\label{s:pf}
	Our main objective in this paper is to address the simultaneous cyber attack and fault detection and isolation (CAFDI) problem for the CPS \eqref{e:e_s} by designing a bank of observers such that each set of residual signals corresponding to  observers is sensitive and specified to detect one specific type of anomaly, namely either an actuator cyber attack $a_\text{u}(t)$, a sensor cyber attack $a_\text{y}(t)$, an actuator fault $f_1(t)$, and/or a pseudo actuator fault $f_2(t)$, while each residual is decoupled from all the other anomalies. 

Decoupling the residuals from one another implies that  occurrence of anomalies only affect those residual signals that are designated to them. We also \underline{do not} limit our focus to detecting only detectable attacks, such as replay attacks. Our goal and objective is  to further detect the so-called undetectable cyber attacks in sense of Definition \ref{def:undetectable}, namely cyber attacks such as covert and zero dynamics. To accomplish our objectives we assume that the adversary cannot compromise all the communication  channels among the proposed plant side and  C\&C side filters, although they have a complete knowledge of  parameters of the filters and detectors.

	\section{Proposed Methodology}\label{s:ObserverD}
	\begin{figure}[!t]
		\centering
		\centerline{\includegraphics[width=\columnwidth]{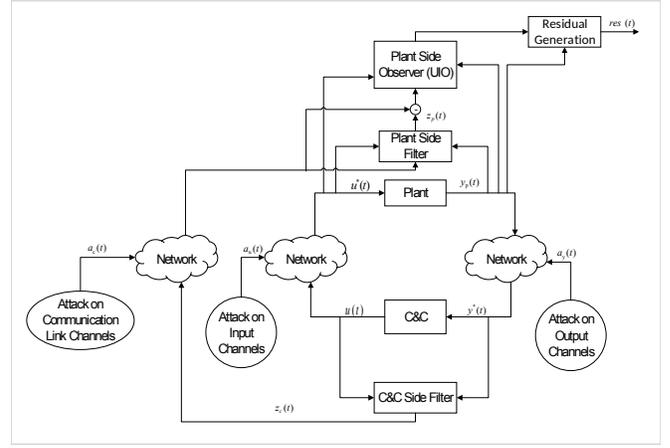}}
		\caption{Observers/filters on both the plant side and the C\&C side of the CPS, where $z_\text{c}(t)$ represents the states of the C\&C side filter, $z_\text{p}(t)$ denotes the states of the plant side filter, $a_\text{c}(t)$ denotes the cyber attack on the communication link channels, and $res(t)$ denotes the residual signals that are generated on the plant side.}\label{fig:distributed}
	\end{figure}
	
	The presence of network layer in CPS has enabled malicious adversaries to perform cyber attacks on the entire system. On the other hand, due to existence of this network layer, it is possible to observe the CPS from both the plant side and its C\&C side. The idea of observing the CPS from both the plant side and the C\&C side is illustrated in Fig.~\ref{fig:distributed}. Our goal in this framework is to utilize information from the designed filters on both sides via a communication link and generate residuals that are specifically sensitive to faults and cyber attacks. Using these residuals, the isolation between faults or cyber attacks can also be achieved.
	
	Two filters having the same characteristics on both sides are designed in Subsections \ref{sub:command_filter} and \ref{sub:pside_filter}. By using the communication link, states of the C\&C side filter are transmitted to the plant side to generate residual signal that is sensitive to only cyber attacks while this communication link may still be compromised by an adversary.

	A detector on the plant side that utilizes an unknown input observer (UIO) is designed in the Subsection \ref{sub_s:UIO}. The detector utilizes the previously generated residuals as  additional input so that they are sensitive to both cyber attacks and faults. The reason for selecting UIO as the main detector is that it enables one to utilize a general design structure to simultaneously address the considered CAFDI problems.

Other algebraic-based observer design techniques, such as eigenstructure assignment and Kalman filters have certain limitations such as not having a flexible structure and requiring high computational cost. For instance, to isolate different types of cyber attacks and faults using Kalman filters, one needs to design and associate a large number of multiple models of Kalman filters on both sides of the CPS, which is computationally excessive and increases the risks and vulnerabilities exploited by intelligent malicious adversaries  to inject cyber attacks.

	Our proposed methodology is presented in the Subsection \ref{sub_s:Pside}. It is worth  noting that by utilizing the proposed methodology, one is still capable of detecting any kind of stealthy cyber attacks on the system, such as covert attacks and zero dynamics attacks.

	\subsection{Command \& Control side filter}\label{sub:command_filter}
	From the C\&C side and according to \eqref{e:e_s}, the output of the CPS is governed by
	\begin{equation}\label{cpovy}
	y^*(t)= Cx(t)+D_{\text{a}}a_\text{y}(t).
	\end{equation}
	We have the following standing assumption to be considered  throughout this paper.
	\begin{assumption}
		Only the communication channels can be compromised and attacked. Consequently, on the C\&C side one has access to the control signal, $u(t)$, before its manipulation by the adversary.
	\end{assumption}
	
	The proposed filter on the C\&C side can be expressed as follows:
	\begin{eqnarray}\label{e:f_c}
	\dot{z}_\text{c}^\ell(t)=F_\text{p}^\ell z_\text{c}^\ell(t) +T_\text{p} B u(t)+K_\text{p}^\ell y^*(t),
	\end{eqnarray}
	where $z_\text{c}^\ell(t) \in \mathbb{R}^{n}$ represents the filter state that estimates $x^\text{s}(t)$ from the C\&C side, and the matrices $F_\text{p}^\ell, \, T_\text{p}^\ell,$ and $K_\text{p}^\ell$ are of appropriate dimensions that are designed and selected subsequently. The index $\ell \in \{\text{SA},\text{AA},\text{SF},\text{AF}\}$, designates if the filter is designed for detecting sensor attacks, actuator attacks, sensor faults, and actuator faults, respectively.

	\subsection{Plant side filter}\label{sub:pside_filter}
	On the plant side, sensor measurements are carried out before sensor attacks, and the output of  CPS can be expressed as follows:
	\begin{equation}\nonumber
	y_\text{p}(t)=Cx(t).
	\end{equation}
	Moreover, on this side one has access to the potentially manipulated control signal $u^*(t)=u(t)+S_\text{a} a_\text{u}(t)$. 

The proposed filter on the plant side is expressed in the following form:
	\begin{equation}\label{e:f_p}
	\begin{split}
	\dot{z}_\text{p}^\ell(t) =& F_\text{p}^\ell z_\text{p}^\ell(t) +T_\text{p}^\ell B u^*(t)+K_\text{p}^\ell y_\text{p}(t)+L_\text{p}^\ell(z_\text{p}^\ell(t) \\
	&-(z_\text{c}^\ell(t)+D_{\text{ac}} a_\text{c}(t))),
	\end{split}
	\end{equation}
	where $z_\text{p}^\ell(t) \in \mathbb{R}^{n}$ denotes the filter state estimating $x^\text{s}(t)$ from the plant side, $a_{\text{c}}(t) \in \mathbb{R}^{n_c}$ denotes the cyber attack on the communication link between the two filters with the signature $D_{\text{ac}}$. Similar to the C\&C side filters, the index $\ell \in \{\text{SA},\text{AA},\text{SF},\text{AF}$\}, indicates if the filter is designed for detecting sensor attacks, actuator attacks, sensor faults, and actuator faults, respectively.
	
	The error signals between  estimated states for both sides can be defined as $e_\text{p}^\ell(t) = z_\text{p}^\ell(t) -z_\text{c}^\ell(t)$. The state-space representation of the error dynamics between the two filter states can be derived as follows:
	\begin{equation}\label{e:e_p}
	\begin{split}
	\dot{e}_\text{p}^\ell(t) = & (F_\text{p}^\ell+L_\text{p}^\ell) e_\text{p}^\ell(t)+T_\text{p}^\ell B_\text{a} a_\text{u}(t) - K_\text{p}^\ell D_\text{a} a_\text{y}(t) \\
	& - L_\text{p}^\ell D_{\text{ac}} a_\text{c} (t) .
	\end{split}
	\end{equation}
	It follows from (\ref{e:e_p}) that the error dynamics  is only sensitive to cyber attacks.

	\subsection{UIO-based detector and residual signal generation}\label{sub_s:UIO}
	Consider a UIO-based detector on the plant side having the following representation:
	\begin{equation}\label{e:UIO}
	\begin{split}
	\dot{z}^\ell(t) =& F^\ell z^\ell(t)+T^\ell B u^*(t)+K^\ell y_\text{p}(t)+L^\ell(z_\text{p}^\ell(t) \\
	& -(z_\text{c}^\ell (t) +D_{\text{ac}} a_\text{c}(t))), \\
	\hat{x}^\ell(t) =& z(t)^\ell+H^\ell y_\text{p}(t),
	\end{split}
	\end{equation}
	\noindent where $z^\ell(t) \in \mathbb{R}^{(n+p_\text{f}+p)}$, and $\hat{x}(t) \in \mathbb{R}^{(n+p_\text{f}+p)}$ denotes the estimated states by the detector. The matrices $F^\ell, \, T^\ell,\, K^\ell, \, L^\ell,$ and $H^\ell$ are of appropriate dimensions and will be specified
 subsequently, with $\ell \in \{\text{SA},\text{AA},\text{SF},\text{AF}$\}, denoting the  categories defined previously.
 %indicates if the detector's output is used  for detecting sensor attacks, actuator attacks, sensor faults, and actuator faults, respectively.

	The error between the states of the detector and the CPS is defined as $e^\ell(t)=x(t)-\hat{x}^\ell(t)$. Let
	\begin{equation}\label{e:res}
	res_\ell(t)=y_\text{p}(t)- C\hat{x}^\ell(t)=C e^\ell(t),
	\end{equation}
	denote a residual signal. By selecting $K^\ell=K_1^\ell+K_2^\ell$, $F^\ell=A-H^\ell CA-K_1^\ell C$, $K_1^\ell$ of appropriate dimension, and $K_2^\ell=FH^\ell$, the dynamics associated with $e^\ell(t)$ can now be expressed in the following form:
	\begin{equation}\label{e:e}
	\begin{split}
	\dot{e}^\ell(t) =& (A-H^\ell CA-K_1^\ell C)e^\ell (t)+(I-T^\ell-H^\ell C)(Bu(t) \\
	& +B_\text{a} a_\text{u}(t))+(I-H^\ell C)F_1 f_1(t) +(I-H^\ell C)F_2 f_2(t) \\
	& +(I-H^\ell C)N \omega (t) -L^\ell e_\text{p}^\ell(t)-L^\ell D_{\text{ac}} a_\text{c}(t).
	\end{split}
	\end{equation}

	\begin{definition}\label{def:res}
		A cyber attack/fault is detected if the residual signal $res_\ell(t)$ given by \eqref{e:res} exceeds a pre-specified threshold $\eta>0$ as follows:
		$$\|res_\ell(t)\|_2 > \eta.$$
		where $\|.\|_2$ indicates the Euclidean norm.
	\end{definition}
	\begin{remark}\label{rem:eta}
		To select the threshold $\eta$, one may need to perform Monte Carlo simulation runs for the healthy system, i.e., for the fault free and cyber attack free system in presence of external disturbances and noise and choose the maximum value of $\|res(t)_\ell\|_2$ as $\eta$.
	\end{remark}
	\begin{definition}[Decoupled Residual] \label{def:decouple}
		The residual signal $res_\ell(t)$ given by \eqref{e:res} is decoupled from  an anomalous  signal in the set $\{a_\text{u}(t), a_\text{y}(t), f_1(t), f_2(t)\} $ if the dynamics and trajectory of $res_\ell(t)$ is not affected by that anomalous signal.
	\end{definition}

	\subsection{Filters and detector design for cyber attack and fault detection and isolation objectives}\label{sub_s:Outerside}\label{sub_s:Pside}
	The error dynamics in  \eqref{e:e_p} and \eqref{e:e} can now be augmented as follows:
	\begin{equation}\label{e:e_trans}
	\begin{split}
	\dot{\check{e}}^\ell(t) =& \check{F}^\ell \check{e}^\ell(t) +\check{B}^\ell u(t) + \check{B}_\text{a}^\ell a_\text{u}(t) + \check{F}_1^\ell f_1(t) + \check{F}_2^\ell f_2(t) \\
	& - \check{K}_\text{p}^\ell a_\text{y}(t) - \check{L}^\ell a_\text{c}(t) +\check{N}^\ell \omega (t),
	\end{split}
	\end{equation}
	where $\check{e}^\ell(t) =[{e^\ell(t)}^\top \, {e_\text{p}^\ell(t)}^\top]^\top$, and
	\begin{equation}\label{e:trans}
	\begin{split}
	\check{F}^\ell  &= \begin{bmatrix}
	F^\ell & -L^\ell \\
	0 & F_\text{p}^\ell + L_\text{p}^\ell
	\end{bmatrix}, \, \check{B}= \begin{bmatrix}
	(I-T^\ell -H^\ell C)B \\
	0
	\end{bmatrix}, \\
	\check{B}_\text{a}^\ell & = \begin{bmatrix}
	(I-T^\ell-H^\ell C)B_\text{a} \\
	T_\text{p}^\ell B_\text{a}
	\end{bmatrix} , \, \check{F}_1^\ell = \begin{bmatrix}
	(I-H^\ell C)F_1 \\
	0
	\end{bmatrix} , \\ \check{F}_2^\ell & = \begin{bmatrix}
	(I-H^\ell C) F_2 \\
	0
	\end{bmatrix} , \,
	\check{K}_\text{p}^\ell = \begin{bmatrix}
	0 \\
	K_\text{p}^\ell D_\text{a}
	\end{bmatrix} , \,
	\check{L}^\ell = \begin{bmatrix}
	L^\ell D_{\text{ac}} \\
	L_\text{p}^\ell D_{\text{ac}}
	\end{bmatrix} , \\
	\check{N}^\ell &= \begin{bmatrix}
	(I-H^\ell C)N \\
	0
	\end{bmatrix},
	\end{split}
	\end{equation}
	where $\ell \in \{\text{SA},\text{AA},\text{SF},\text{AF}\}$.
	\begin{assumption}\label{assume:known}
		The malicious adversary is aware of the parameters of filters in (\ref{e:f_c}), (\ref{e:f_p}), and the UIO-based detector in (\ref{e:UIO}).
	\end{assumption}
	\begin{assumption}\label{assume:Dac}
		The malicious  adversary does not have access to \underline{all} the communication  channels between the two side filters, i.e., $\text{rank}(D_\text{ac}) < n$.
	\end{assumption}

	In the following, it is shown that how one can generate four residual signals $res_\text{AA}(t)$, $res_\text{SA}(t)$, $res_\text{AF}(t)$, and $res_\text{SA}(t)$ to detect the actuator cyber attack, the sensor cyber attack, the actuator fault, and the sensor fault, respectively, by using a bank of filters and four UIO-based detectors.

	\begin{proposition}\label{theorm:a_a}
		Under Assumption \ref{assume:Dac}, the residual signal $res_\text{AA}(t)=y_\text{p}(t)- C\hat{x}^\text{AA}(t)$ is affected by the \underline{actuator cyber attack} $a_\text{u}(t)$ and is decoupled from $a_\text{y}(t)$, $f_1(t)$,  and $f_2(t)$ in the sense of Definition \ref{def:decouple}, if the following conditions for the augmented dynamics \eqref{e:e_trans} hold for $\ell=\text{AA}$, namely:
		\begin{enumerate}
			\item $T^\ell = I-H^\ell C$;
			\item $(I-H^\ell C)F_1 = 0$;
			\item $(I-H^\ell C)F_2 = 0$;
			\item $L^\ell D_{\text{ac}} = 0$;			\item $L_\text{p}^\ell D_{\text{ac}} = 0$;
			\item $K_\text{p}^\text{AA} D_\text{a} = 0$;
			\item the triplet $(C, \, F^\ell, \, L^\ell)$ is left-invertible;
			\item the Rosenbrock system matrix
			\begin{equation*}
			P_{\Sigma_\text{u}}(s) =	\begin{bmatrix}
			sI-(F_\text{p}^\text{AA}+L_\text{p}^\text{AA}) & -T_\text{p}^\text{AA} B_\text{a} \\
			L^\text{AA} & 0_{(n+p_\text{f}+p) \times m_\text{a}}
			\end{bmatrix},
			\end{equation*}
			does not have any non-minimum phase zero dynamics;
			\item $	\text{rank}\, (L^\text{AA} T_\text{p}^\text{AA} B_\text{a}) = \text{rank}\, (T_\text{p}^\text{AA} B_\text{a})$;
			\item $\check{F}^\ell$ is Hurwitz.		
		\end{enumerate}
	\end{proposition}
	\begin{proof}
		The augmented governing error dynamics associated with $e^\text{AA}(t)$ and $e_p^\text{AA}(t)$ are governed by \eqref{e:e_trans} where $\ell = \text{AA}$.
		Under Conditions 1) to 6), the dynamics \eqref{e:e_trans} become
		\begin{equation}
		\dot{\check{e}}^\text{AA}(t) = \check{F}^\text{AA} \check{e}^\text{AA}(t) + \check{B}_\text{a}^\text{AA} a_\text{u}(t)+\check{N}^\text{AA}\omega (t). \label{res_AA}
		\end{equation}
		Consequently, the error signal ${\check{e}}(t)$ is not affected by the control command $u(t)$, the actuator fault $f_1(t)$, the sensor fault $f_2(t)$, the sensor attack $a_\text{y}(t)$, and the communication link attack signal $a_\text{c}(t)$.
		Furthermore, \eqref{res_AA} can be partitioned into the following two subsystems:
		\begin{align}\label{e:th_e_p}
		\dot{e}_\text{p}^\text{AA}(t)=(F_\text{p}^\text{AA} + L_\text{p}^\text{AA}){e}_\text{p}^\text{AA}(t)+T_\text{p}^\text{AA} B_\text{a} a_\text{u}(t),
		\end{align}
		and
		\begin{align}\label{e:th_e}
		\begin{split}
		\dot{{e}}^\text{AA}(t)=&F e^\text{AA}(t)-L^\text{AA}{e}_\text{p}^\text{AA}(t)+(I-H^\text{AA} C)N\omega(t),\\
		res_\text{AA}(t)=&Ce^\text{AA}(t).
		\end{split}
		\end{align}
		
		Based on Condition 7) and according to \eqref{e:th_e}, the impact of ${e}_\text{p}^\text{AA}(t)$ will appear in $res_\text{AA}(t)$ for any $a_\text{u}(t) \neq 0$.
		
		Consider $e_\text{p}^\text{AA}(t)$ in \eqref{e:th_e_p} with the output $L^\text{AA}e_\text{p}^\text{AA}(t)$ in order to construct the Rosenbrock system matrix $P_{\Sigma_\text{u}}(s)$. To prevent stealthy attacks on the plant side filter, one needs to design this filter and $L^\text{AA}$ such that the Rosenbrock system matrix $P_{\Sigma_\text{u}} (s)$ has no non-minimum phase zero dynamics and is left-invertible \cite{safcpsuua}.
		
		The Rosenbrock system matrix $P_{\Sigma_\text{u}}(s)$ being left-invertible is equivalent to the largest controllability subspace of the system $(L^\text{AA}, F_\text{p}^\text{AA} + L_\text{p}^\text{AA}, T_\text{p}^\text{AA} B_\text{a})$ contained in $\text{ker} (L^\text{AA})$, and designated as $\mathscr{R}^*(\Sigma_\text{u})$ being null \cite{ctfls}. One has (refer to Theorem 8.22 in \cite{ctfls} and Theorem 5.6 in \cite{wonham1974linear})
		\begin{equation}\label{eq:R_largestContSubspace}
		\mathscr{R}^*(\Sigma_\text{u}) = \mathscr{V}(\Sigma_\text{u}) \cap \mathscr{W}^*(\Sigma_\text{u}),
		\end{equation}
		where $\mathscr{V}(\Sigma_\text{u})$ is the weakly unobservable subspace that is equivalent to the largest output-nulling subspace of the triplet $(L^\text{AA}, F_\text{p}^\text{AA} + L_\text{p}^\text{AA}, T_\text{p}^\text{AA} B_\text{a})$, and $\mathscr{W}^*(\Sigma_\text{u})$ is the smallest conditioned invariant subspace containing $\text{Im}(T_\text{p}^\text{AA} B_\text{a})$ \cite{siolcps}.
		
		As described in \cite{ctfls} and \cite{wonham1974linear}, these subspaces can be computed by using the following algorithm
		\begin{IEEEeqnarray}{rCl}\label{e:th5:v}
			\mathscr{V}_0 &=& \text{Ker} (L^\text{AA}), \nonumber \\
			\mathscr{V}_k &=& \mathscr{V}_0 \cap {F_\text{p}^\text{AA}}^{-1} (\mathscr{V}_{k-1} + \text{Im}(T_\text{p}^\text{AA} B_\text{a})),
		\end{IEEEeqnarray}
		and	\begin{IEEEeqnarray}{rCl}\label{e:th5:w}
			\mathscr{W}_0 &=& \text{Im} (T_\text{p}^\text{AA} B_\text{a}), \nonumber \\
			\mathscr{W}_k &=& \mathscr{W}_0 + F_\text{p}^\text{AA} (\mathscr{W}_{k-1} \cap \text{Ker} (L^\text{AA})),
		\end{IEEEeqnarray}
		where  $\mathscr{V}_k$ and $\mathscr{W}_k$ converge to $\mathscr{V}(\Sigma_u)$ and $\mathscr{W}^*(\Sigma_u)$, respectively, in at most $k=n$ steps.
		
		Given \eqref{eq:R_largestContSubspace},  $\mathscr{R}^*(\Sigma_\text{u})=0$, if $\mathscr{V}_0 \cap \mathscr{W}_0 =0$, or equivalently,
		\begin{equation}\label{e:th5:intersectionCondition}
		\text{Ker} (L^\text{AA}) \cap \text{Im} (T_\text{p}^\text{AA} B_\text{a}) =0.
		\end{equation}
		The equation \eqref{e:th5:intersectionCondition} implies that $\text{Im} (T_\text{p}^\text{AA} B_\text{a})$ should not be in the null space of $L^\text{AA}$, which is equivalent to
		$$\text{rank}\, (L^\text{AA}T_\text{p}^\text{AA} B_\text{a}) = \text{rank}\, (T_\text{p}^\text{AA} B_\text{a}).$$
		
		The Rosenbrock system matrix $P_{\Sigma_\text{u}}(s)$ being left-invertible implies that for any $a_\text{u}(t)\neq 0$, $L^\text{AA}e_\text{p}^\text{AA}(t) \neq 0$.
		
		Finally, in order to detect actuator cyber attacks, the governing dynamics in \eqref{res_AA} should be stable. This completes the proof of the Proposition 1.
	\end{proof}

	\begin{remark}
		It should be emphasized that as per Assumption \ref{assume:Dac}, there exists a nonzero $L^\text{AA}$ that satisfies the Condition (4) in the above proposition.
	\end{remark}
	\begin{proposition} \label{theorm:s_a}
		% 		Let $T^\text{SA}=T^\text{AA}$, $H^\text{SA}=H^\text{AA}$, $L^\text{SA}=L^\text{AA}$, $L_\text{p}^\text{SA}=L_\text{p}^\text{AA}$, and $\check{F}^\text{SA}=\check{F}^\text{AA}$.
		Under Assumption \ref{assume:Dac}, the residual signal $res_\text{SA}(t)=y_\text{p}(t)- C\hat{x}^\text{SA}(t)$ is affected by the \underline{sensor cyber attacks} $a_\text{y}(t)$ and is decoupled from $a_\text{u}(t)$, $f_1(t)$,  and $f_2(t)$ in the sense of Definition \ref{def:decouple},  if Conditions 1)-5), 7), and 10) of the Proposition \ref{theorm:a_a} for $\ell=\text{SA}$, and the following conditions for the augmented error dynamics  \eqref{e:e_trans} hold:
		\begin{enumerate}
			\item $T_\text{p}^\text{SA} B_\text{a} = 0$;
			\item the Rosenbrock system matrix \begin{equation*}
			P_{\Sigma_\text{y}}(s) =	\begin{bmatrix}
			sI-(F_\text{p}^\text{SA}+L_\text{p}^\text{SA}) & K_\text{p}^\text{SA} D_\text{a} \\
			L^\text{SA} & 0_{(n+p_\text{f}+p) \times p_\text{a}}
			\end{bmatrix},
			\end{equation*}
			does not have any non-minimum phase zero dynamics; and
			\item $\text{rank}\, (L^\text{SA} K_\text{p}^\text{SA} D_\text{a}) = \text{rank}\, (K_\text{p}^\text{SA} D_\text{a})$.
		\end{enumerate}
	\end{proposition}
	\begin{proof}
		The proof follows along similar lines to that of Proposition~\ref{theorm:a_a} and is omitted for  sake of brevity.
	\end{proof}
	
	\begin{remark}\label{rem:au_ac}
		Suppose Condition (9) of the Proposition \ref{theorm:a_a} is not satisfied and $P_{\Sigma_\text{u}}(s)$ is not left-invertible. In this case, it has been shown in \cite{safcpsuua} that one can find an actuator cyber attack $a_\text{u}(t) \neq 0$ such that $L^\text{AA}e_\text{p}^\text{AA} (t)=0$. This type of cyber attack has been represented in \cite{safcpsuua} and has been defined as ``undetectable controllable attack" in \cite{siolcps}. According to \eqref{e:th_e_p} and \eqref{e:th_e} the  actuator cyber attack signal $a_\text{u}(t)$ can affect the error $e^\text{AA}(t)$ only through $L^\text{AA}e_\text{p}^\text{AA}(t)$. Hence, the adversary has the capability of injecting a stealthy cyber attack by using $a_\text{u}(t)$ that does not affect the residual signal $res_\text{AA}(t)=C e^\text{AA}(t)$. Similarly, it can be shown that if Condition (3) of Proposition \ref{theorm:s_a} is not satisfied and $P_{\Sigma_\text{y}}(s)$ is not left-invertible, the adversary can inject stealthy attack using $a_\text{y}(t)$ which does not affect the residual $res_\text{SA}(t)$.
	\end{remark}
	
	\begin{remark}\label{rem:4}
		In Propositions  \ref{theorm:a_a} and \ref{theorm:s_a}, there is no assumption on the nature, characteristics, and type of sensor and actuator cyber attacks. This implies that by using the proposed method, one is capable of detecting and isolating detectable attacks, such as replay attacks, as well as undetectable attacks (refer to Definition \ref{def:undetectable}), such as covert attacks and zero dynamics attacks.
	\end{remark}

	\begin{proposition}\label{theorm:a_f}
		Let $\ell= \text{AF}$ . The residual signal $res_\text{AF}(t)=y_\text{p}(t)- C\hat{x}^\text{AF}(t)$ is affected by the \underline{actuator fault} $f_1(t)$ and is decoupled from $a_\text{u}(t)$, $a_\text{y}(t)$,  and $f_2(t)$ in the sense of Definition \ref{def:decouple},  if $L^\text{AF} = 0$ and the following conditions hold:
		\begin{enumerate}
			\item $T^\text{AF} = I-H^\text{AF} C$;
			\item $(I-H^\text{AF} C)F_2 = 0$;
			\item $\check{F}^\text{AF}$ is Hurwitz.		
		\end{enumerate}
	\end{proposition}
	\begin{proof}
		In light of Conditions 1) and 2), and setting $\ell = \text{AA}$, \eqref{e:e_trans} yields
		\begin{align*}
		\dot{\check{e}}^\text{AF}(t) =& \check{F}^\text{AF} \check{e}^\text{AF}(t) + \check{B}_\text{a}^\text{AF} a_\text{u}(t) + \check{F}_1^\text{AF} f_1(t) - \check{K}_\text{p}^\text{AF} a_\text{y}(t) \\
		&- \check{L}^\text{AF} a_\text{c}(t) +\check{N}^\text{AF}\omega(t).
		\end{align*}
		
		Moreover, by setting $L^\text{AF}=0$, the dynamics of ${e}^\text{AF}(t)$ is governed by:
		\begin{equation*}
		\dot{{e}}^\text{AF}(t) = {F}^\text{AF} {e}^\text{AF}(t) + (I-H^\text{AF} C){F}_1 f_1(t)+{N} \omega(t).
		\end{equation*}
		and consequently,  the residual signal $res_\text{AF}(t)=C{e}^\text{AF}(t)$  is only sensitive to the actuator fault $f_1(t)$. In addition, $\check{F}^\text{AF}$ should be Hurwitz in order to have a stable error dynamics ${e}^\text{AF}(t)$. This completes the proof of the Proposition 3.
	\end{proof}
	\begin{proposition}\label{theorm:s_f}
		The residual signal $res_\text{SF}(t)=y_\text{p}(t)- C\hat{x}^\text{SF}(t)$ is affected by the \underline{pseudo actuator fault} $f_2(t)$ and is decoupled from $a_\text{u}(t)$, $a_\text{y}(t)$,  and $f_1(t)$ in the sense of Definition \ref{def:decouple}, if $L^\text{SF}=0$ and the following conditions for the augmented dynamic \eqref{e:e_trans} hold:
		\begin{enumerate}
			\item $T^\text{SF} = I-H^\text{SF} C$;
			\item $(I-H^\text{SF} C)F_1 = 0$;
			\item $\check{F}^\text{SF}$ is Hurwitz.		
		\end{enumerate}
	\end{proposition}
	\begin{proof}
		Setting $\ell= \text{SF}$, the proof follows along similar lines to that of Proposition \ref{theorm:a_f} and is omitted for  sake of brevity.
	\end{proof}

	As stated in \cite{douioarfdf}, Conditions 2) and 3) in  Proposition \ref{theorm:a_a} are solvable if and only if
	$\text{rank}(CF_1)=\text{rank}(F_1);$
	and $\text{rank}(CF_2)=\text{rank}(F_2).$
	The next lemma provides sufficient conditions for isolability of sensors and actuator faults.
	\begin{theorem}\label{th:UIO}
		The  residuals $res_\text{AF}(t)$ and $res_\text{SF}(t)$ can be simultaneously generated to detect and isolate $f_1(t)$ and $f_2(t)$ if
		$F_1^\top F_2=0$.
	\end{theorem}
	\begin{proof}
		In order to generate the residual signal $res_\text{AF}(t)$ Condition 2) in  Proposition \ref{theorm:a_f} should hold,  which can be interpreted as requiring
		\begin{equation}\label{e:lem_iso1}
		\text{Im}(I-H^\text{AF}C) \subset \text{Ker}(F_2^\top).
		\end{equation}
		and at the same time, the impact of $f_1(t)$ should show up in the dynamics of $e(t)$, that implies $(I-H^\text{AF}C)F_1 \neq 0$. The latter condition is equivalent to
		\begin{equation}\label{e:lem_iso2}
		\text{Im}(F_1^\top) \subset \text{Im}(I-H^\text{AF}C).
		\end{equation}
		From \eqref{e:lem_iso1} and \eqref{e:lem_iso2}, it can be inferred that $\text{Im}(F_1^\top) \subset \text{Ker}(F_2^\top)$, which implies that $F_1^\top F_2=0$. Note that the case of generating the residual signal  $res_\text{SF}(t)$ provides one with the same result. This completes the proof of the Theorem 1.
	\end{proof}
	
It follows from the definitions of $F_1$ and $F_2$ that the condition $F_1^\top F_2=0$ is always satisfied. Therefore, as long as Conditions (2) and (3) in  Proposition \ref{theorm:a_a} are solvable, the actuator faults and pseudo actuator faults can be detected and isolated.
	
	\begin{remark}
		To generate the residual signals $res_\text{AA}(t)$, $res_\text{SA}(t)$, $res_\text{AF}(t)$, and $res_\text{SF}(t)$ one needs to construct a bank of eight filters (four on each side) with the states $z_\text{p}^\text{AA}(t)$, $z_\text{c}^\text{AA}(t)$, $z_\text{p}^\text{SA}(t)$, $z_\text{c}^\text{SA}(t)$, $z_\text{p}^\text{AF}(t)$, $z_\text{c}^\text{AF}(t)$, $z_\text{p}^\text{SF}(t)$, and $z_\text{c}^\text{SF}(t)$ and four UIO-based detectors with the states $\hat{x}^\text{AA}(t)$, $\hat{x}^\text{SA}(t)$, $\hat{x}^\text{AF}(t)$, and $\hat{x}^\text{SF}(t)$ according to Propositions \ref{theorm:a_a}-\ref{theorm:s_f}. In Propositions \ref{theorm:a_a} and \ref{theorm:s_a}, the matrices $K_\text{p}^\text{AA}$ and $T_\text{p}^\text{SA}$ have been utilized to decouple sensor cyber attacks and actuator cyber attacks in  sense of Definition \ref{def:decouple} from the generated residual signals, respectively. Hence, one can conclude that there is no contradiction among the conditions to generate $res_\text{AA}(t)$ and $res_\text{SA}(t)$. Subsequently, from Theorem \ref{th:UIO} it can be seen that no contradiction exists among the design conditions in the Propositions \ref{theorm:a_f} and \ref{theorm:s_f} to generate $res_\text{AF}(t)$ and $res_\text{SF}(t)$. Moreover, in Propositions \ref{theorm:a_f} and \ref{theorm:s_f}, the matrix $L^\ell$ has been employed to decouple the cyber attack signals from $res_\text{AF}(t)$ and $res_\text{SF}(t)$, which indicates that there are no contradictions in the design conditions of Propositions \ref{theorm:a_a} and \ref{theorm:s_a}.
	\end{remark}

	\section{Numerical Case Studies}\label{s:simo}
	In this section, numerical case studies are provided to demonstrate and verify the capabilities and advantages of our proposed methodology as compared to the available results in the literature. For these case studies, a bank of filters and UIO-based detectors are designed to achieve detection and isolation of cyber attacks as well as faults by using the proposed methods in the Propositions \ref{theorm:a_a}-\ref{theorm:s_f}. To simulate the covert and zero dynamics attacks the models in \cite{adaii} and \cite{ascf} are used, respectively.
	
Two types of cyber attacks are studied, namely covert attacks and zero dynamics attacks. Moreover, detection and isolation of \underline{simultaneous} actuator and sensor bias faults with  cyber attacks are also demonstrated and validated. A linear dynamical system with the following characteristic matrices and cyber attack and fault signatures is considered:
	\begin{align*}
	A^\text{s} &= \begin{bmatrix}
	-1 & 0 & 1 & 0 \\
	0 & -3 & 0 & 1 \\
	0 & 0 & -2 & 0 \\
	0 & 0 & 0 & -2
	\end{bmatrix},  \, B^\text{s}= \begin{bmatrix}
	-2 & -1 \\
	0 & -2 \\
	0 & -3 \\
	-4 & 0
	\end{bmatrix},  \\
	C^\text{s} &= \begin{bmatrix}
	0.2 & 0 & 0 & 0 \\
	0 & 0.2 & 0 & 0
	\end{bmatrix}, \,
	B_\text{a}^\text{s} = \begin{bmatrix}
	-2 & -1 \\
	0 & -2 \\
	0 & -3 \\
	-4 & 0
	\end{bmatrix},   \\
	L_1 &= \begin{bmatrix}
	-2 \\
	0 \\
	0 \\
	-4
	\end{bmatrix}, \, L_2^\text{a} = \begin{bmatrix}
	1\\
	0 \\
	0
	\end{bmatrix},
	D_{\text{ac}} = \begin{bmatrix}
	1 & 0 & 0 & 0 \\
	0 & 1 & 0 & 0 \\
	0 & 0 & 0 & 0 \\
	0 & 0 & 0 & 0
	\end{bmatrix}, \, N^\text{a}= \begin{bmatrix}
	0 \\
	1 \\
	1
	\end{bmatrix},  \\
	A^\text{a} &= \begin{bmatrix}
	-1 & 0 & 0 \\
	0 & -2 & 0 \\
	0 & 0 & -3 \\
	\end{bmatrix}, \, C^\text{a}= \begin{bmatrix}
	1 & 1 & 1 \\
	1 & 1 & 1
	\end{bmatrix}, \, N^\text{s}= \begin{bmatrix}
	1 \\
	1 \\
	1 \\
	1
	\end{bmatrix}, \\ D_\text{a} &= \begin{bmatrix}
	0.2 & 0\\
	0 & 0.2
	\end{bmatrix},
	\stepcounter{equation}\tag{\theequation}\label{e:exmpl1:sys}
	% 	\tag{\stepcounter{equation}\theequation\label{e:exmpl1:sys}}
	\end{align*}
	where all the input and output channels are compromised by adversaries as they have access to \underline{two}  out of the \underline{four} communication  channels. The covariance matrices of $\omega^\text{s}(t)$ and $\omega^\text{a} (t)$ are specified  as $Q=\text{diag}(0.01,\, 0.01, \, 0.01, \, 0.01)$ and $R^\text{a}=\text{diag}(0.02,\, 0.02)$, respectively. 

For the case studies, the design steps that are summarized in the Algorithms \ref{alg:woc} and \ref{alg:f} in Appendix \ref{s:apndx} are utilized. A bank of plant side filters as given by \eqref{e:f_p}, C\&C side filters as presented by \eqref{e:f_c}, and  detectors as provided in \eqref{e:UIO} are designed such that the conditions of  Propositions \ref{theorm:a_a}-\ref{theorm:s_f} are satisfied. Moreover, the residual signals $res_\text{AA}(t)$, $res_\text{SA}(t)$, $res_\text{AF}(t)$, and $res_\text{SF}(t)$ are generated according to Propositions \ref{theorm:a_a}-\ref{theorm:s_f}, respectively.

	\noindent \textbf{Scenario 1 (Zero Dynamics Attacks)}: The system presented in \eqref{e:exmpl1:sys} has a non-minimum phase zero at $s=0.3028$, that is associated with the zero state direction $x_0^\text{s}=[0, \, 0, \, -0.6514, \, 1]^\top$ and the zero input direction $u_0=[-0.5757 ,\, 0.5]^\top$. To determine the threshold for the residual signals $res_\text{AA}(t)$ and $res_\text{SA}(t)$ of the actuator and sensor cyber attacks 100 Monte Carlo simulation runs are conducted according to Remark \ref{rem:eta}, and the threshold is determined as $\eta=3.3$. The parameters of the filters and the UIO-based detector subject to actuator cyber attack are designed as follows:
	\begin{align*}
	F_\text{p}^\text{AA} &= \begin{bmatrix}
	-3 & 0 & 0 & 0\\
	0 & -2 & 0 & 0\\
	0 & 0 & -4 & 0 \\
	0 & 0 & 0 & -5
	\end{bmatrix},  \,
	T_\text{p}^\text{AA} = \begin{bmatrix}
	1 & 1 & 1 & 1 \\
	1 & 2 & 3 & 1 \\
	2 & 0 & 0 & -4 \\
	0 & 1 & 0 & 0
	\end{bmatrix}, \\
	L_\text{p}^\text{AA} &= \begin{bmatrix}
	0 & 0 & 4 & -1 \\
	0 & 0 & 3 & -2 \\
	0 & 0 & 2 & -3 \\
	0 & 0 & 5 & -1
	\end{bmatrix}, \, H^\text{AA}= \begin{bmatrix}
	5 & -5 \\
	0 & 0 \\
	0 & 0 \\
	10 & -10 \\
	0 & 1 \\
	0 & 0 \\
	0 & 0
	\end{bmatrix}, \\
	K_1^\text{AA} &= \begin{bmatrix}
	6 & -2 \\
	-3 & 1 \\
	6 & 2 \\
	3 & 1 \\
	3 & 1 \\
	6 & 2 \\
	3 & 1
	\end{bmatrix}, \,  L^\text{AA} = \begin{bmatrix}
	0 & 0 & 4 & -1 \\
	0 & 0 & 3 & -2 \\
	0 & 0 & 2 & -3 \\
	0 & 0 & 5 & -1 \\
	0 & 0 & 3 & -2 \\
	0 & 0 & 2 & -3 \\
	0 & 0 & 5 & -1
	\end{bmatrix}, \, K_\text{p}^\text{AA} = [0]_{4 \times 2},
	\end{align*}
	
	As can be seen in Fig.~\ref{fig:exp1:zero}, the residual signal $res_\text{AA}(t)=y_\text{p}(t)-C\hat{x}^\text{AA}(t)$ that is designed to detect  actuator cyber attacks has increased (due to a zero dynamics attack) while the other residuals are successfully below the threshold.
	\begin{figure}
		\centering
		\centerline{\includegraphics[width=\columnwidth]{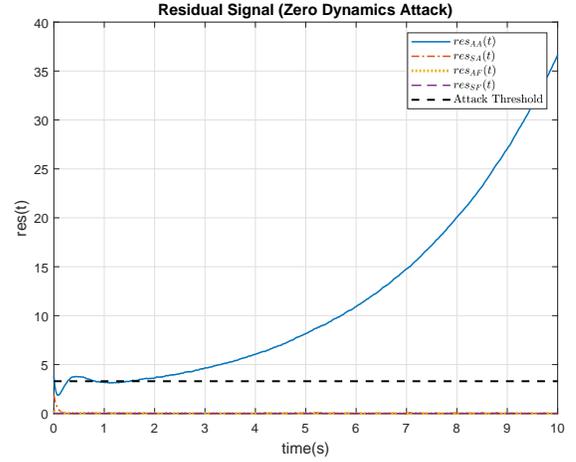}}
		\caption{Detection of a zero dynamics attack that is injected at $t=0$ (s).}\label{fig:exp1:zero}
	\end{figure}
	
	\noindent \textbf{\bf Scenario 2 (Covert Attacks):} In this scenario, a covert attack scenario is considered. The adversary is capable of completely removing the impact of actuator cyber attack $a_\text{u}(t)=[2, \, 1]^\top$ from the sensor measurements  by using the sensor cyber attack $D_{\text{a}} a_\text{y}(t)=- C x_{\text{cov}}(t)$, where $\dot{x}_\text{cov} (t) = A x_\text{cov}(t) + B_\text{a} a_\text{u}(t)$ and $x_{\text{cov}}(0)=x(0)$. The impact of this cyber attack at $t=10$ (s) can be seen on  sensor measurements on the plant side as shown in Fig.~\ref{fig:exp1:y_cov}. However, the received sensor measurements on the C\&C side do not show any anomaly in outputs. The parameters of the detector are the same as in Scenario 1, but to detect  sensor cyber attacks a set of filters are designed to satisfy the conditions that are provided in Proposition \ref{theorm:s_a} to generate $res_\text{SA}(t)$.
	
	\begin{figure}[!t]
		\centering
		\centerline{\includegraphics[width=\columnwidth]{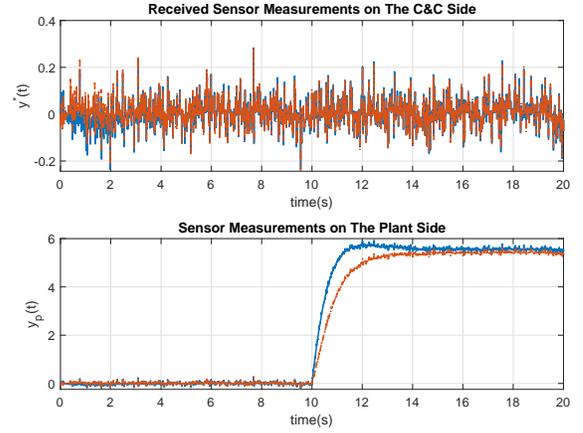}}
		\caption{Difference between  output of the system on the plant side and the C\&C side due to injection of covert attack at $t=10$ (s).}\label{fig:exp1:y_cov}
	\end{figure}
	
	As shown in Fig.~\ref{fig:exp1:cov}, the increase in  actuator and sensor cyber attacks residuals, $res_\text{AA}(t)$ and $res_\text{SA}(t)$, respectively, that exceed the threshold indicate the occurrence of these cyber attacks.
	
	\begin{figure}[!t]
		\centering
		\centerline{\includegraphics[width=\columnwidth]{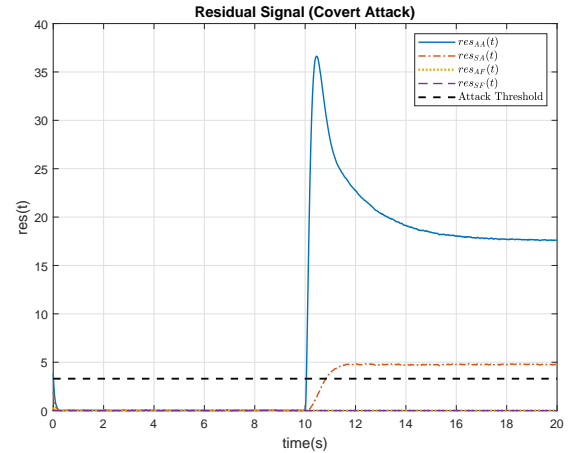}}
		\caption{Detection of actuator and sensor cyber attacks in case of covert attacks.}\label{fig:exp1:cov}
	\end{figure}
	
	\noindent \textbf{Scenario 3 (Faults)}: Using Proposition \ref{theorm:a_f}, the UIO-based detector and its corresponding residual signal $res_\text{AF}(t)$ that is sensitive to  actuator faults are first designed. Then, based on conditions in Proposition \ref{theorm:s_f} to detect  sensor faults the matrices for the UIO-based detector and the residual signal $res_\text{SF}(t)$ are selected. The threshold for residuals that are used to detect actuator and sensor faults is computed according to the method provided in Remark \ref{rem:eta} and is set to $\eta=0.6$. In this scenario, the actuator fault, $f_1(t)=40$, has occurred at $t=5$ (s) and the pseudo actuator fault, $f_2(t)=20$, also exists in the system from $t=10$ (s) onwards. It can be observed from Fig.~\ref{fig:exp1:fault} that due to occurrence of faults the corresponding residuals have been increased.
	
	\begin{figure}[!t]
		\centering
		\centerline{\includegraphics[width=\columnwidth]{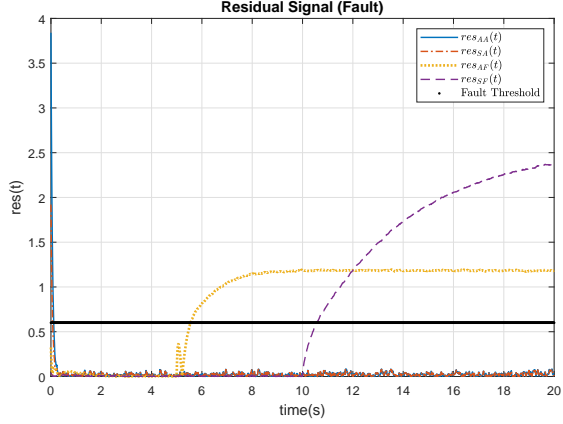}}
		\caption{Detection of actuator and sensor faults.}\label{fig:exp1:fault}
	\end{figure}

	\noindent \textbf{Scenario 4 (Simultaneous Injection of Cyber Attack and Fault)}: In this scenario,  the detection and isolation of simultaneous cyber attacks and faults is demonstrated. In this scenario, the system is under a covert attack at $t=0$ (s) and an actuator fault and sensor faults occur at $t=5$ (s) and $t=10$ (s), respectively. As depicted in Fig.~\ref{fig:exp1:attack&fault},  these anomalies can be both detected and isolated successfully.

	\begin{figure}[!t]
		\centering
		\centerline{\includegraphics[width=\columnwidth]{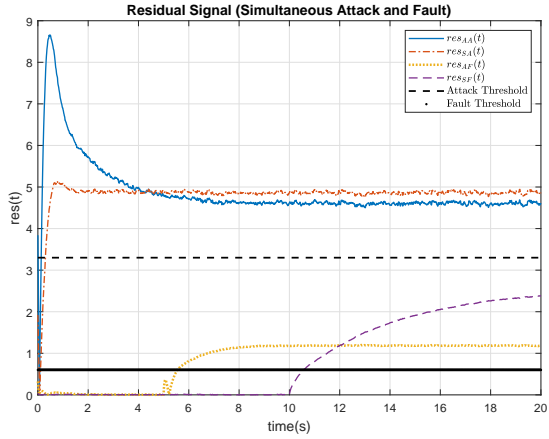}}
		\caption{Detection and isolation of different simultaneous cyber attacks and faults.}\label{fig:exp1:attack&fault}
	\end{figure}

	\noindent \textbf{Scenario 5 (Condition (9) of the Proposition \ref{theorm:a_a} is not Satisfied)}: In this scenario, we have intentionally designed our monitoring system in a manner such that Condition (9) of the Proposition \ref{theorm:a_a} is not satisfied. Therefore, we can illustrate its importance in our proposed methodology. In Fig.~\ref{fig:exp2:WOCoding}, it can be seen that when the above condition is \underline{not satisfied} the adversary is now capable of performing ``undetectable controllable attack" (refer to Remark \ref{rem:au_ac} and \cite{siolcps}) on $P_{\Sigma_\text{u}}(s)$ and completely eliminate or cancel out  impacts of the actuator cyber attack on the residual.
	
	\begin{figure}[!t]
		\centering
		\centerline{\includegraphics[width=\columnwidth]{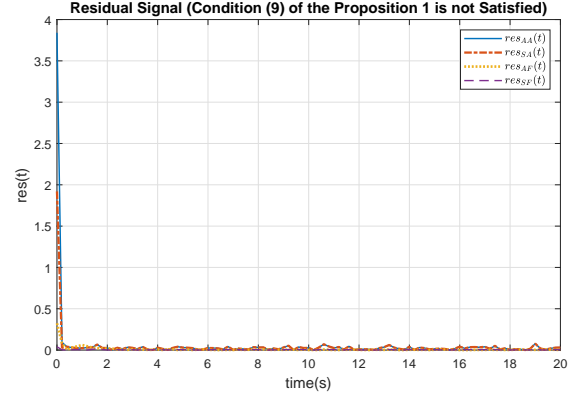}}
		\caption{Residual signals when Condition (9) of the Proposition \ref{theorm:a_a} is not satisfied.}\label{fig:exp2:WOCoding}
	\end{figure}

	\noindent \textbf{Comparative Study with Results Available  in the Literature}: In order to provide a comparison with the existing results in the literature, the proposed approach in \cite{docaazdaicps} is applied to our case studies. The following periodic modulation matrix was developed in \cite{docaazdaicps}:
	\begin{equation*}
	S(k)=\systeme{S_1 \, \, ; k=1 \, \, (0\leq t \leq t_1),
		\, \, \vdots,
		S_T \, \, ; k=T \, \, (t_{T-1}\leq t \leq t_T)}
	\end{equation*}
	where $S(k)$ is the modulation matrix on the input, $S_1, \hdots, S_T \in \mathbb{R}^{m \times m}$ are constant matrices, and $T=m$. The idea in \cite{docaazdaicps} is to disrupt the knowledge of the adversary from the system by employing the modulation $S(k)$. Using the detection method in \cite{docaazdaicps}, it is shown in Fig.~\ref{fig:exp1:comp} that despite having no actuator and sensor cyber attacks, the attack residual signal increases which \textit{misleadingly} indicates the existence of cyber attacks (false positive). However, in the same figure it is shown that by using our proposed method in Propositions \ref{theorm:a_a}-\ref{theorm:s_f} and generating $res_\text{AA}(t)$, $res_\text{SA}(t)$, $res_\text{AF}(t)$, and $res_\text{SF}(t)$, the occurrence of actuator fault in the system was correctly detected and isolated.
	
	\begin{figure}[!t]
		\centering
		\centerline{\includegraphics[width=\columnwidth]{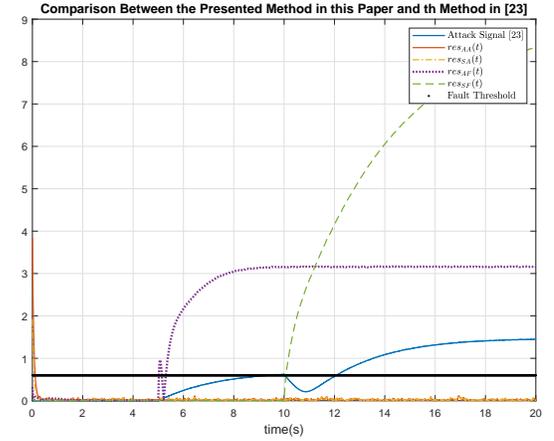}}
		\caption{False detection of cyber attack by using the proposed method in \cite{docaazdaicps} while there is only a fault in the system (actuator fault is injected from $t=5$ (s) onwards).}\label{fig:exp1:comp}
	\end{figure}

	\subsection{Quantitative performance evaluation}
	Our proposed CAFDI methodologies under different noise levels are quantitatively evaluated through 100 different Monte Carlo simulation runs. A confusion matrix \cite{FAWCETT2006861} is employed to evaluate the performance of our proposed methods. Given a classifier and its corresponding instances, four possible outcomes are specified as (1) TP (True Positive), if the instance is positive and is truly classified as positive, (2) FN (False Negative), if the instance is positive and incorrectly classified as negative, (3) TN (True Negative), if the instance is negative and correctly classified as negative, and (4) FP (False Positive), if the instance is negative and incorrectly classified as positive \cite{FAWCETT2006861}. 

Based on the possible outcomes the metric true positive rate (TPR) which indicates the rate of correct detection is used as a performance measure in this paper. This performance measure can be computed by using the expression $TPR=TP/(TP+FN)$. In this subsection, ``AA", ``SA", ``AF", and ``SF" are used to denote actuator attack, sensor attack, actuator fault, and sensor fault, respectively. The TPR results for the proposed methods that are developed in Propositions \ref{theorm:a_a}-\ref{theorm:s_f} for injection of cyber attacks and faults are presented in Tables \ref{tb:TPR_AA}-\ref{tb:TPR_SF}. 

The rows in Table \ref{tb:TPR_AA} indicate the TPR of   actuator attack (AA) detection given different scenarios for simultaneous occurrence of anomalies in the system, such as occurrence of AA and SA, AA and AF, and AA and SF. Furthermore, the second column in this table shows the computed TPR for Proposition \ref{theorm:a_a}. In Table \ref{tb:TPR_SA} the rows show the TPR of   sensor attack (SA) detection in various scenarios for simultaneous occurrence of anomalies in the system. Moreover, the second column corresponds to the computed TPR of detection for SA where Proposition \ref{theorm:s_a} is utilized. In Table \ref{tb:TPR_AF}, the computed TPR of detection of actuator fault  (AF) in presence of different anomalies are shown in the rows. Finally, the rows in Table \ref{tb:TPR_SF} indicate the TPR for sensor fault (SF) under simultaneous occurrences of  anomalies in the system.
	\begin{table}
		\caption{TPR measure for actuator attack detection according to the proposed methodology corresponding to Proposition \ref{theorm:a_a}.  }
		\begin{center}\label{tb:TPR_AA}
			\begin{tabular}{ | c | c |}
				\hline
				Types of Anomalies & TPR\% (Proposition \ref{theorm:a_a}) \\ \hline
				AA & 96\% \\ \hline
				AA \& SA & 96\% \\ \hline
				AA \& AF & 95\% \\ \hline
				AA \& SF & 96\% \\ \hline
				AA \& SA \& AF & 95\% \\ \hline
				AA \& SA \& SF & 96\% \\ \hline
				AA \& AF \& SF & 96\% \\ \hline
				AA \& SA \& AF \& SF & 95\% \\ \hline
			\end{tabular}
		\end{center}
	\end{table}

	\begin{table}
		\caption{TPR measure for sensor attack detection according to the proposed methodology corresponding to Proposition \ref{theorm:s_a}.}
		\begin{center}\label{tb:TPR_SA}
			\begin{tabular}{ | c | c |}
				\hline
				Types of Anomalies & TPR\% (Proposition \ref{theorm:s_a}) \\ \hline
				SA & 99\% \\ \hline
				SA \& AA & 99\% \\ \hline
				SA \& AF & 99\% \\ \hline
				SA \& SF & 99\% \\ \hline
				SA \& AA \& AF & 99\% \\ \hline
				SA \& AA \& SF & 99\% \\ \hline
				SA \& AF \& SF & 99\% \\ \hline
				SA \& AA \& AF \& SF & 90\% \\ \hline
			\end{tabular}
		\end{center}
	\end{table}

	\begin{table}
		\caption{TPR measure for actuator fault detection according to the proposed methodology in Proposition \ref{theorm:a_f}.}
		\begin{center}\label{tb:TPR_AF}
			\begin{tabular}{ | c | c |}
				\hline
				Types of Anomalies & TPR\% (Proposition \ref{theorm:a_f}) \\ \hline
				AF & 93\% \\ \hline
				AF \& AA & 93\% \\ \hline
				AF \& SA & 93\% \\ \hline
				AF \& SF & 93\% \\ \hline
				AF \& AA \& SA & 93\% \\ \hline
				AF \& AA \& SF & 93\% \\ \hline
				AF \& SA \& SF & 92\%  \\ \hline
				AF \& AA \& SA \& SF & 93\% \\ \hline
			\end{tabular}
		\end{center}
	\end{table}

	\begin{table}
		\caption{TPR measure for sensor fault detection according to the proposed methodology in Proposition \ref{theorm:s_f}.}
		\begin{center}\label{tb:TPR_SF}
			\begin{tabular}{ | c | c |}
				\hline
				Types of Anomalies & TPR\% (Proposition \ref{theorm:s_f}) \\ \hline
				SF & 96\% \\ \hline
				SF \& AA & 96\%  \\ \hline
				SF \& SA & 96\%  \\ \hline
				SF \& AF & 96\%  \\ \hline
				SF \& AA \& SA & 96\%  \\ \hline
				SF \& AA \& AF & 96\% \\ \hline
				SF \& SA \& AF & 96\% \\ \hline
				SF \& AA \& SA \& AF & 96\% \\ \hline
			\end{tabular}
		\end{center}
	\end{table}
	
	\section{Conclusion}\label{s:conclu}
	In this paper, the problem of simultaneous detection and isolation of machine induced faults and intelligent malicious adversarial cyber attacks has been studied. A methodology based on the cyber-physical systems (CPS) two side filters and a UIO-based detector has been proposed. In this method, a filter was designed on the plant side with its dynamics different from the C\&C side filter so that even if the adversary estimates the parameters of the C\&C side filter they cannot identify the parameters of the plant side filter. Moreover, this methodology inhibits adversaries from disguising their cyber attacks. Using the proposed strategy, one is capable of \textit{simultaneously} detecting machine induced actuator and sensor faults as well as undetectable cyber attacks, such as covert and zero dynamics attacks, and detectable cyber attacks, such as the replay attack. In future work we will consider non-ideal communication networks. Furthermore, to make the cyber-physical systems model closer to the real-world applications, we will extend the results of this paper to a multi-agent based framework.
	
	\section{Appendix}\label{s:apndx}

	\begin{algorithm}
		\caption{Pseudo code for cyber attack detection based on  Propositions \ref{theorm:a_a} and \ref{theorm:s_a}.}\label{alg:woc}
		\begin{enumerate}
			\item[] \textbf{UIO-based detector design:}
			\item Find $H^\text{AA}$ such that $(I-H^\text{AA} C)F_1 = 0$ and $(I-H^\text{AA} C)F_2 = 0$.
			\item Compute $K_1^\text{AA}$ such that $F^\text{AA}=A-H^\text{AA} CA-K_1C$ is Hurwitz.
			\item Set $T^\text{AA}=I-H^\text{AA} C$.
			\item Find $L^\text{AA}$ such that $L^\text{AA} D_{\text{ac}}=0$ and check if the Rosenbrock system matrix $$\begin{bmatrix}
			sI-F^\text{AA} & L^\text{AA} \\
			C & 0_{p \times n}
			\end{bmatrix}$$ is left-invertible, if not go to Step 1 where $H^\text{AA}$, $K_1^\text{AA}$, and $L^\text{AA}$ are changed.
			\item[] \textbf{Design of filters and residual generation subject to actuator cyber attack detection (Proposition \ref{theorm:a_a}):}
			\item Find $K_\text{p}^\text{AA}$ such that $K_\text{p}^\text{AA} D_\text{a}=0$.
			\item Compute $L_\text{p}^\text{AA}$ such that $L_\text{p}^\text{AA} D_{\text{ac}}=0$.
			\item Find a diagonal matrix $F_\text{p}^\text{AA}$ and the matrix $T_\text{p}^\text{AA}$ such that the Rosenbrock system matrix $$P_{\Sigma_\text{u}}(s) =\begin{bmatrix}
			sI-(F_\text{p}^\text{AA}+L_\text{p}^\text{AA}) & -T_\text{p}^\text{AA}B_\text{a} \\
			L^\text{AA} & 0_{(n+p_\text{f}+p) \times m_\text{a}}
			\end{bmatrix}$$ does not have any non-minimum phase zero dynamics and $\text{rank}\, (L^\text{AA}T_\text{p}^\text{AA} B_\text{a}) = \text{rank}\, (T_\text{p}^\text{AA} B_\text{a})$.
			\item Check if $(F_\text{p}^\text{AA}+L_\text{p}^\text{AA})$ is Hurwitz, if not go to Step (6).
			\item Generate the residual signal $res_\text{AA}(t)$ and compute the threshold $\eta_\text{AA}$ according to Remark \ref{rem:eta}.
			\item[] \textbf{Design of filters and residual generation subject to sensor cyber attack detection (Proposition \ref{theorm:s_a}):}
			\item Set $T^\text{SA}=T^\text{AA}$, $H^\text{SA}=H^\text{AA}$, $L^\text{SA}=L^\text{AA}$, and ${F}^\text{SA}={F}^\text{AA}$.
			\item Find $T_\text{p}^\text{SA}$ such that $T_\text{p}^\text{SA} B_\text{a}=0$.
			\item Compute $L_\text{p}^\text{SA}$ such that $L_\text{p}^\text{SA} D_{\text{ac}}=0$.
			\item Find a diagonal matrix $F_\text{p}^\text{SA}$ and the matrix $K_\text{p}^\text{SA}$ such that the Rosenbrock system matrix $$P_{\Sigma_\text{y}}(s) =\begin{bmatrix}
			sI-(F_\text{p}^\text{SA}+L_\text{p}^\text{SA}) & K_\text{p}^\text{SA} D_\text{a} \\
			L^\text{SA} & 0_{(n+p_\text{f}+p) \times p_\text{a}}
			\end{bmatrix}$$ does not have any non-minimum phase zero dynamics and $\text{rank}\, (L^\text{SA}K_\text{p}^\text{SA} D_\text{a}) = \text{rank}\, (K_\text{p}^\text{SA} D_\text{a})$.
			\item Generate the residual signal $res_\text{SA}(t)$ and compute the threshold $\eta$ according to Remark \ref{rem:eta}.
		\end{enumerate}
	\end{algorithm}

	\begin{algorithm}
		\caption{Pseudo code for fault detection based on  Propositions \ref{theorm:a_f} and \ref{theorm:s_f}.}\label{alg:f}
		\begin{enumerate}
			\item[] \textbf{UIO-based detector design and residual generation subject to actuator fault detection (Proposition \ref{theorm:a_f}):}
			\item Find $H^\text{AF}$ such that $(I-H^\text{AF}C)F_2 = 0$.
			\item Compute $K_1^\text{AF}$ such that $F^\text{AF}=A-H^\text{AF}CA-K_1^\text{AF}C$ is Hurwitz.
			\item Set $T^\text{AF}=I-H^\text{AF}C$.
			\item Set $L^\text{AF}=0$.
			\item Generate the residual signal $res_\text{AF}(t)$ and compute the threshold $\eta$ according to Remark \ref{rem:eta}.
			\item[] \textbf{UIO-based detector design and residual generation subject to sensor fault detection (Proposition \ref{theorm:s_f}):}
			\item Find $H^\text{SF}$ such that $(I-H^\text{SF}C)F_1 = 0$.
			\item Set $T^\text{SF}=I-H^\text{SF}C$ and $L^\text{SF}=L^\text{AF}$.
			\item Generate the residual signal $res_\text{SF}(t)$ and compute the threshold $\eta$ according to Remark \ref{rem:eta}.
		\end{enumerate}
	\end{algorithm}
	
	\bibliographystyle{IEEEtran}
	\bibliography{CAFDIRef}
	
\end{document}